\documentclass[12pt,reqno]{amsart}

\newcommand\version{April 08, 2022}

\usepackage{amsmath,amsfonts,amsthm,amssymb,amsxtra}
\usepackage{hyperref}
\usepackage{color}
\usepackage{bbm} 
\usepackage{bm} 

\newtheorem{theorem}{Theorem}[section]
\newtheorem{proposition}[theorem]{Proposition}

\newtheorem{corollary}[theorem]{Corollary}

\theoremstyle{definition}

\newtheorem{assumption}[theorem]{Assumption}

\theoremstyle{remark}

\newtheorem{remark}[theorem]{Remark}
\newtheorem{remarks}[theorem]{Remarks}

\numberwithin{equation}{section}

\newcommand{\cb}{\mathcal{B}}

\renewcommand{\epsilon}{\varepsilon}

\newcommand{\F}{\mathcal{F}}

\newcommand{\MV}{\mathrm{MV}}
\newcommand{\N}{\mathbb{N}}

\renewcommand{\phi}{\varphi}
\newcommand{\R}{\mathbb{R}}

\newcommand{\T}{\mathbb{T}}

\newcommand{\Z}{\mathbb{Z}}

\DeclareMathOperator{\supp}{supp}
\DeclareMathOperator{\sgn}{sgn}

\DeclareMathOperator{\tr}{Tr}

\def\bs{\mathbb{S}}

\def\ch{\mathcal{H}}

\def\co{\mathcal{O}}

\def\cs{\mathcal{S}}
\def\cv{\mathcal{V}}

\newcommand{\me}[1]{\mathrm{e}^{#1}}
\newcommand{\one}{\mathbf{1}}

\makeatletter
\newcommand*{\rom}[1]{\expandafter\@slowromancap\romannumeral #1@}
\makeatother

\begin{document}

\title[Number and sums of eigenvalues --- \version]{On the number and sums of eigenvalues of Schr\"odinger-type operators with degenerate kinetic energy}

\author[J.-C. Cuenin]{Jean-Claude Cuenin}
\address[Jean-Claude Cuenin]{Department of Mathematical Sciences, Loughborough University, Loughborough, Leicestershire, LE11 3TU United Kingdom}
\email{J.Cuenin@lboro.ac.uk}

\author[K. Merz]{Konstantin Merz}
\address[Konstantin Merz]{Institut f\"ur Analysis und Algebra, Technische Universit\"at Braunschweig, Universit\"atsplatz 2, 38106 Braunschweig, Germany}
\email{k.merz@tu-bs.de}


\date{\version}

\begin{abstract}
  We estimate sums of functions of negative eigenvalues of Schr\"odinger-type
  operators whose kinetic energy vanishes on a codimension one submanifold.
  Our main technical tool is the Stein--Tomas theorem and some of its generalizations.
\end{abstract}

\dedicatory{Dedicated to the memory of Sergey N.~Naboko}
\maketitle

\section{Introduction}

For $d\geq1$ we consider Schr\"odinger-type operators of the form
\begin{align}
  \label{eq:defh}
  H = T(-i\nabla)-V \quad \text{in}\ L^2(X^d)\,,\quad X\in\{\R,\Z\}
\end{align}
where the kinetic energy $T(\xi)$ vanishes
on a codimension-one submanifold.
A prime example is $T=|\Delta+1|$, which naturally appears, e.g., in the
BCS theory of superconductivity and superfluidity, see, e.g.,
Frank, Hainzl, Naboko, and Seiringer \cite{Franketal2007T},
Hainzl, Hamza, Seiringer, and Solovej \cite{Hainzletal2008T},
and Hainzl and Seiringer \cite{HainzlSeiringer2008C,HainzlSeiringer2008S},
as well as the Hartree--Fock theory of the electron gas (jellium), see, e.g.,
Gontier, Hainzl, and Lewin \cite{Gontieretal2019}.
The potential $V$ is assumed to be real-valued and sufficiently regular,
so that $H$ can be realized as self-adjoint operator. 
In this note we are interested in estimates for sums of functions of negative
eigenvalues of $H$ when $V\in L^q$ for some $q<\infty$.
We now state our assumptions on $T$.

\begin{assumption}
  \label{assumt}
  Assume that $T(\xi)\geq0$ attains its minimum on a smooth compact codimension
  one submanifold $S = \{\xi\in\R^d:T(\xi)=0\}$.
  Assume that there exists an open, precompact neighborhood $\Omega\subseteq\R^d$
  of $S$ such that the following holds.
  \begin{enumerate}
  \item There exists $P\in C^\infty(\Omega)$ such that $T(\xi)=|P(\xi)|$.
    Let $\tau:=\max_{\xi\in\Omega}T(\xi)$.
    
  \item There exist $c_P>0$ such that $|\nabla P(\xi)|\geq c_P$ for all
    $\xi\in\Omega$. 

  \item There exist constants $C_1,C_2>0$ and $s\in(0,d)$ such
    that $T(\xi)\geq C_1|\xi|^s+C_2$ for $\xi\in\R^d\setminus\Omega$.
  \end{enumerate}
  For $t>0$, consider the level set
  $S_t:=\{\xi\in\R^d:\ |P(\xi)|=t\}$ which is a smooth compact codimension one
  submanifold embedded in $\R^d$ with corresponding surface measure $d\Sigma_{S_t}$.
  We set $d\sigma_{S_t}(\xi):=d\Sigma_{S_t}(\xi)/|\nabla P(\xi)|$
  and assume that
  \begin{enumerate}
  \item[(4)] there is $r>0$ such that
    $\sup_{t\in(0,\tau)}|(d\sigma_{S_t})^\vee(x)|\lesssim_\tau (1+|x|)^{-r}$,
    where
    $(d\sigma_{S_t})^\vee(x)=\int_{S_t}\me{2\pi ix\cdot\xi}\,d\sigma_{S_t}(\xi)$
    denotes the Fourier transform of $d\sigma_{S_t}$. 
  \end{enumerate}
\end{assumption}

Assumptions (1)-(3) also appear in the work of Hainzl and Seiringer
\cite{HainzlSeiringer2010}, where it is assumed that $V\in L^1\cap L^{d/s}$.
These assumptions imply that the quadratic form
$\langle u,(T(-i\nabla)-V)u\rangle$ is bounded from below, whenever
$u\in C_c^\infty(\R^d)$. 
The Friedrichs extension then provides us with a self-adjoint operator
$H=T(-i\nabla)-V$. Note that the constants $\tau,c_P,C_1,C_2$ in Assumption
\ref{assumt} are fixed $\co(1)$-quantities.

Assumption (4) is related to the curvature of $S_t$ and is crucial since it
allows us to consider $V\in L^q\cap L_{\rm loc}^{d/s}$ with $q>1$.
Littman \cite{Littman1963} showed that if
$S$ has $2r\in\{0,1,...,d-1\}$ non-vanishing principles curvatures, then
one has the decay $|(d\sigma_S)^\vee(x)|=\co(|x|^{-r})$.
In particular, Assumption (4) holds for $T=|\Delta+1|$ with $r=(d-1)/2$.
Note also that this assumption is always guaranteed in the nonzero
curvature case, whenever one has the decay
$(d\sigma_{S})^\vee(x) = \co(|x|^{-(d-1)/2})$ for $t=0$.
(See, e.g., \cite[Proposition~4.1]{CueninMerz2021}.)

\medskip
For $V\in L^q$ with $q\in[d/s,\infty)$ the essential spectrum
$\sigma_{\rm ess}(H)=[0,\infty)$ coincides with that of $T(-i\nabla)$.
The discrete spectrum of the operator $H_\lambda:=T(-i\nabla)-\lambda V$
for $0<\lambda\ll1$ has recently received considerable interest.
For $V\in L^1\cap L^{d/s}(\R^d)$ it has been shown, e.g., by
Frank, Hainzl, Naboko, and Seiringer \cite{Franketal2007T}
and Hainzl and Seiringer \cite{HainzlSeiringer2008C,HainzlSeiringer2010}
that for any eigenvalue $a_S^j>0$ of the operator
\begin{align}
  \label{eq:defvsintro}
  \begin{split}
    L^2(S,d\sigma_S) & \to L^2(S,d\sigma_S)\,,\\
    u & \mapsto \int_S \hat V(\xi-\eta)u(\eta)d\sigma_S(\eta)\,, \quad u \in L^2(S,d\sigma_S)\,,
  \end{split}
\end{align}
there is a corresponding eigenvalue $-e_j(\lambda)<0$ of $T-\lambda V$
which satisfies
\begin{align}
  \label{eq:weakcouplingintro}
  e_j(\lambda) = \exp\left(-\frac{1}{2\lambda a_S^j}(1+o(1))\right)
  \quad \text{as}\ \lambda \to 0\,.
\end{align}
Here, $\hat V(\xi)=\int_{\R^d}\me{-2\pi ix\cdot\xi}V(x)\,dx$
denotes the Fourier transform of $V$ in $\R^d$.
Recently, the authors \cite{CueninMerz2021} extended this result to a
substantially larger class of potentials, such as
$V\in L^q(\R^d)$ with $q\in[d/s,r+1]$, whenever $T(-i\nabla)$
satisfies also the curvature assumption (4) with $r+1\geq d/s$
in Assumption \ref{assumt}.
This is clearly the case for $T=|\Delta+1|$ with $r=(d-1)/2$.

On the other hand, Laptev, Safronov, and Weidl \cite{Laptevetal2002}
studied the asymptotic behavior of the eigenvalues $-e_j<0$ of $T-V$
as $j\to\infty$, when $V$ is of the form $V(x) = v(x)(1+|x|)^{-1-\epsilon}$,
where $v\in L^\infty(\R^d)$ satisfies $v(x)=w(x/|x|)(1+o(1))$ as
$|x|\to\infty$ with $w\in C^\infty(\bs^{d-1})$. Similarly as in
\eqref{eq:weakcouplingintro}, the eigenvalue asymptotics is determined by
that of the eigenvalues $a_S^j>0$ of the operator in \eqref{eq:defvsintro}.
Their main result \cite[Theorem 4.4]{Laptevetal2002} essentially relied
on an abstract theorem (Theorem 3.4 there) which connected the spectral
asymptotics of $H$ and \eqref{eq:defvsintro} with each other. In turn,
the limit $\lim_{j\to\infty}a_S^j$ is well understood thanks to the works
\cite{BirmanSolomjak1977E} of Birman and Solomjak on singular values
of (asymptotically) homogeneous pseudodifferential operators with symbol
$h_1(x)a(x,\xi)h_2(x)$. Here $h_1,h_2\in C_c^\infty$, and
$a(x,t\xi)=t^{-\beta}a(x,\xi)$ for all $|\xi|\geq1$ and $t>1$.
By a change of coordinates, the operator in \eqref{eq:defvsintro} can
be transformed into this operator modulo ``error operators'' which do
not change the leading order of the spectral asymptotics of
\eqref{eq:defvsintro}. We refer to Birman and Yafaev
\cite{BirmanYafaev1984} for a detailed exposition
and for the explicit expression for $\lim_{j\to\infty}a_S^j$.
For $V$ merely in $L^q(\R^d)$, the results of Birman and Solomjak
are not applicable.
It would be interesting to study the asymptotics $\lim_{j\to\infty}e_j$ in this case.

\medskip
The purpose of this note is to prove estimates for sums of functions $f(x)$
on $\R_+$ of the absolute values of the negative eigenvalues of $T-V$ when
$V\in L^q$.
For $f(x)=x^\gamma$ this will lead to modifications of the celebrated
Lieb--Thirring inequality
\cite{LiebThirring1975,LiebThirring1976,Lieb1976B}
\begin{align}
  \label{eq:classicltintro}
  \tr[(-\Delta-V)_-]^\gamma
  \leq c_{d,\gamma} \int_{\R^d} V(x)_+^{\gamma+\frac{d}{2}}\,dx
\end{align}
with $\gamma\geq1/2$ if $d=1$, $\gamma>0$ if $d=2$, and $\gamma\geq0$
if $d\geq 3$, and a constant $c_{d,\gamma}>0$ which is independent of $V$.
Here we denote the positive and negative parts of a real number or a
self-adjoint operator by $X_+:=\max\{X,0\}$ and $X_-:=\max\{-X,0\}$,
respectively.
We refer to Frank \cite{Frank2021} for a recent review of its history,
applications, and generalizations.
Observe that the right side of \eqref{eq:classicltintro} is homogeneous
in $V$.
Since the assumptions on $T(\xi)$, i.e., the constants $\tau,c_P,C_1,C_2$
appearing in Assumption \ref{assumt} are fixed $\co(1)$-quantities,
we do not expect scale-invariant inequalities.

Nevertheless, non-scale-invariant inequalities relating
sums of eigenvalues with $L^q$-norms of $V$, such as
Daubechies' inequality \cite{Daubechies1983}
\begin{align}
  \label{eq:daubechies}
  & \tr[(\sqrt{-\Delta+1}-1-V)_-]
    \leq c_d \int_{\R^d} \left(V(x)_+^{1+\frac{d}{2}} + V(x)_+^{1+d}\right)\,dx
\end{align}
for the pseudorelativistic Chandrasekhar operator,
are important in the analysis of many-particle quantum systems.
In fact, using the techniques of \cite{Lieb1976B}, Daubechies extended
\eqref{eq:daubechies} to a larger class of operators $T(-i\nabla)$.
However, these results are not applicable in the present situation,
since they require $T(\xi)$ to be a spherically symmetric and strictly
increasing function with $T(0)=0$. This condition is not satisfied by
the operators $T$ we consider here, such as $T=|\Delta+1|$.
Further examples of eigenvalue estimates involving a sum of two terms
were proved, e.g., by Lieb, Solovej, and Yngvason \cite{Liebetal1994S}
in the context of the Pauli operator and
by Exner and Weidl \cite{ExnerWeidl2001} in the context of Schr\"odinger
operators in wave guides $\omega\times\R$ with $\omega\subset\R^{d-1}$.
For two-term estimates for eigenvalue sums of Schr\"odinger operators on
metric trees, we refer to Frank and Kova\v{r}\'{\i}k \cite[Theorem~6.1]{FrankKovarik2013},
see also Ekholm, Frank, and Kova\v{r}\'{\i}k \cite{Ekholmetal2011},
Molchanov and Vainberg \cite{MolchanovVainberg2010},
and the references therein for further results.
Finally, we refer to Frank, Lewin, Lieb, and Seiringer \cite{Franketal2013}
for two-term estimates for eigenvalue sums of Schr\"odinger operators in
presence of a constant positive background density.

Besides sums of powers of eigenvalues, we also prove estimates for
sums of powers of logarithms (i.e., $f(x)=(\log(2+1/x))^{-\gamma}$) of
eigenvalues of $T-V$. This is natural, as \eqref{eq:weakcouplingintro}
indicates that the eigenvalues of $T-V$ cluster with an exponential rate
at zero.
In particular, the proofs of these results yield estimates on the
eigenvalues $e_j$ and show how fast they cluster at zero as
$j\to\infty$, see \eqref{eq:bounden2}.
However, we do not investigate the asymptotics for
$\lim_{j\to\infty}e_j$ here.
The idea of deriving estimates for logarithms of eigenvalues is not new
and has already been considered by Kova{\v{r}}\'\i k, Vugalter, and Weidl
\cite{Kovariketal2007} in the context of two-dimensional Schr\"odinger
operators $-\Delta-V$, whose eigenvalues also cluster exponentially fast
at the bottom of the essential spectrum, see Simon \cite{Simon1976}.

If $T$ degenerates sublinearly, we are able to prove
Cwikel--Lieb--Rosenbljum-type estimates
\cite{Cwikel1977,Lieb1976B,Rosenbljum1972} for the number of
negative eigenvalues. We illustrate this using $T=|\Delta+1|^{1/s}$
with $s>1$.

Finally, we generalize our results to lattice Schr\"odinger-type operators
on $\ell^2(\Z^d)$. Under the same curvature assumption we obtain
better estimates than in $L^2(\R^d)$ due to the absence of high energies.

\subsection*{Organization and notation}
In Section \ref{s:preliminaries} we collect facts about Schatten spaces and
Fourier restriction theory that are used in the subsequent sections.
In Section \ref{s:classiclt} we prove estimates for the number of
eigenvalues of $T-V$ in $L^2(\R^d)$ below a fixed threshold $-e<0$
(Theorem \ref{tracerestriction}).
Then we prove inequalities for sums of powers of eigenvalues 
(Theorem \ref{classiclt}),
and for sums of powers of logarithms of eigenvalues of $T-V$ (Theorem \ref{loglt}).
We conclude with a Cwikel--Lieb--Rosenbljum bound
for $|\Delta+1|^{1/\sigma}-V$ with $\sigma>1$ (Theorem \ref{clrbcs}).
In Section \ref{s:discretesetting} we consider the corresponding problems for
Schr\"odinger operators on $\ell^2(\Z^d)$.
We first recall two versions of a discrete Laplace operator and a modification
of the ``BCS operator'' $|\Delta+1|-V$ to $\ell^2(\Z^d)$.
In Section \ref{s:ltlattice} we prove estimates on the number of
negative eigenvalues of $T-V$ in $\ell^2(\Z^d)$ below a threshold
$-e<0$ (Theorem \ref{tracerestrictionlattice}), ordinary and
logarithmic Lieb--Thirring-type inequalities (Theorems
\ref{classicltlattice} and \ref{logltlattice}), and a
Cwikel--Lieb--Rosenbljum bound for powers of the modified BCS operator
in $\ell^2(\Z^d)$ (Theorem \ref{fracbcsclrlattice}).

We write $A\lesssim B$ for two non-negative quantities $A,B\geq0$ to
indicate that there is a constant $C>0$ such that $A\leq C B$.
If $C=C_\tau$ depends on a parameter $\tau$, we write $A\lesssim_\tau B$.
The dependence on fixed parameters like $d$ and $s$ is sometimes
omitted. Constants are allowed to change from line to line.
The notation $A\sim B$ means $A\lesssim B\lesssim A$.
All constants are denoted by $c$ or $C$ and are allowed to change
from line to line.
We abbreviate $A\wedge B:=\min\{A,B\}$ and $A\vee B:=\max\{A,B\}$.
The Heaviside function is denoted by $\theta(x)$. We use the
convention $\theta(0)=1$.
The indicator function and the Lebesgue measure of a set
$\Omega\subseteq\R^d$ are denoted by $\one_\Omega$ and $|\Omega|$,
respectively.
For $x\in\R^d$ we write $\langle x\rangle:=(2+x^2)^{1/2}$.

\section{Preliminaries}
\label{s:preliminaries}

\subsection{Trace ideals}
We collect some facts on trace ideals that are used in this note, 
see also, e.g., Birman--Solomjak \cite[Chapter~11]{BirmanSolomjak1987}
or Simon \cite{Simon2005}.

Let $(\cb,\|\cdot\|)$ denote the Banach space of all linear, bounded
operators on a Hilbert space $\ch$.
The $p$-th Schatten space of all compact operators $T\in\cs^\infty(\ch)$
whose singular values $\{s_n(T)\}_{n\in\N}$ (in non-increasing order,
appearing according to their multiplicities) satisfy
$\|T\|_{\cs^p(\ch)}^p:=\sum_{n\geq1}s_n(T)^p<\infty$ for $p>0$ is
denoted by $\cs^p(\ch)$.
We denote the $p$-th weak Schatten space over $\ch$ by
\begin{align}
  \label{eq:defweakschatten}
  \cs^{p,\infty}(\ch) := \{T\in\cs^\infty(\ch):\, \|T\|_{\cs^{p,\infty}(\ch)}^p := \sup_{\lambda>0}\lambda^p n(\lambda,T)<\infty\} \supseteq\cs^p(\ch)\,,
\end{align}
where
\begin{align}
  \label{eq:defnlambdat}
  n(\lambda,T) := \#\{n:\,s_n(T)>\lambda\}\,, \quad \lambda>0\,.
\end{align}
Note that
\begin{align}
  \|T\|_{\cs^{p,\infty}(\ch)} = \sup_m s_m(T)m^{\frac1p}\,,
\end{align}
which together with \eqref{eq:defweakschatten} implies in particular
\begin{align}
  \label{eq:evboundtraceideal}
  s_m(T) \leq \|T\|_{\cs^{p,\infty}(\ch)}\,m^{-\frac1p}
  \quad \text{and} \quad
  n(\lambda,T) \leq \|T\|_{\cs^{p,\infty}(\ch)}^p\lambda^{-p}\,.
\end{align}

If $T:\ch\to\ch'$ is a linear operator between two Hilbert spaces $\ch$
and $\ch'$ we denote its $p$-th Schatten norm by $\|T\|_{\cs^p(\ch,\ch')}$.
If $\ch=\ch'$, we either write
$\|T\|_{\cs^{p}(\ch)}$, $\|T\|_{\cs^{p}}$, or $\|T\|_{p}$, and abbreviate
$\cs^p(\ch)=\cs^p$.
Analogous notation is used for $\cs^{p,\infty}$.

\subsection{Fourier restriction and extension}

Let $X\in\{\R,\Z\}$, $\hat X=\R$ when $X=\R$,
and $\hat X=\T$ when $X=\Z$, where $\T^d:=(\R/\Z)^d$ denotes
the $d$-dimensional torus with Brillouin zone $[-1/2,1/2)^d$.
If $X=\Z$, then the $L^q(X^d)$-spaces are equipped with counting measure
so that $L^q(\Z^d)\equiv\ell^q(\Z^d)$ for any $q>0$.

Let $S$ be a smooth, compact codimension one submanifold embedded in
$\hat X^d$ with induced Lebesgue surface measure $d\Sigma_S$.
If $S$ is the level set of a smooth real-valued function
$P\in C^\infty(\hat X^d)$, i.e.,
$S=\{\xi\in\hat X^d:\,P(\xi)=0\}$, then the Leray measure
\cite{GelfandShilov1964} is $d\sigma_S(\xi)=|\nabla P(\xi)|^{-1}d\Sigma_S(\xi)$.
We introduce the Fourier restriction operator
\begin{align}
  F_{S}:\cs(X^d)\to L^2(S,d\sigma_S)\,,
  \qquad \phi\mapsto (F_S\phi)(\xi) = \widehat{\phi}(\xi)\big|_S
  = \int_{X^d}\me{-2\pi ix\cdot\xi}\phi(x)\,dx \big|_S
\end{align}
and its adjoint, the Fourier extension operator
\begin{align}
  F_S^*:L^2(S,d\sigma_S)\to\cs'(X^d)\,,
  \quad u \mapsto (F_S^* u)(x) = \int_S u(\xi) \me{2\pi ix\cdot\xi} \,d\sigma_S(\xi)\,.
\end{align}
Under the additional assumption that the Gaussian curvature of $S$ is
non-zero everywhere, the Stein--Tomas theorem
\cite{Tomas1975,Stein1986,Bourgain2003}
asserts that $F_S:L^p(X^d)\to L^2(S)$ is bounded for all
$p\in[1,2(d+1)/(d+3)]$.
Its proof relies on the bound
$|(d\sigma_S)^\vee(x)|\lesssim\langle x\rangle^{-\frac{d-1}{2}}$.
By duality, the Stein--Tomas theorem is equivalent to the operator norm
bound
$\|W_1F_S^*F_SW_2\|_{L^2(X^d)\to L^2(X^d)}\lesssim \|W_1\|_{L^{2q}(X^d)}\|W_2\|_{L^{2q}(X^d)}$
for all $W_1,W_2\in L^{2q}$, whenever $1/q=1/p-1/p'$ and $p\in[1,2(d+1)/(d+3)]$,
i.e., $q\in[1,(d+1)/2]$.
Frank and Sabin \cite[Theorem 2]{FrankSabin2017} upgraded this to
a Schatten norm estimate.
For smooth compact hypersurfaces $S\subseteq \hat X^d$ with everywhere
non-vanishing Gaussian curvature and
\begin{align}
  \label{eq:defsigmaq}
  \sigma(q) := \frac{(d-1)q}{d-q}\,, \qquad q\in[1,d)\,,
\end{align}
Frank and Sabin proved
\begin{align}
  \label{eq:tsschatten}
  \|W_1F_S^*F_SW_2\|_{\cs^{\sigma(q)}(L^2(X^d))}
  \lesssim_{d,S,q} \|W_1\|_{L^{2q}(X^d)}\|W_2\|_{L^{2q}(X^d)}\,,
  \quad q\in\left[1,\frac{d+1}{2}\right]\,.
\end{align}
Note that $\sigma(1)=1$, $\sigma((d+1)/2)=d+1$, and $\sigma(q)\geq q$.

\medskip
As discussed in the introduction, if
$S$ has $2r\in\{0,1,...,d-1\}$ non-vanishing principle curvatures, then
one has the weaker decay $|(d\sigma_S)^\vee(x)|\lesssim\langle x\rangle^{-r}$,
which, as Greenleaf \cite{Greenleaf1981} showed, implies that
$F_S:L^p(X^d)\to L^2(S)$ is bounded for all $p\in[1,(2+2r)/(2+r)]$.
For a given decay rate of $|(d\sigma_S)^\vee(x)|$ the first author proved the
following generalization of \eqref{eq:tsschatten}. 

\begin{proposition}[{\cite[Proposition~A.5]{Cuenin2019}}]
  \label{franksabingen}
  Let $S\subseteq \hat X^d$ be a smooth compact hypersurface with normalized
  defining function\footnote{This means that $S=\{P=0\}$ and $|\nabla P|=1$ on $S$.}
  $P:\hat X^d\to\R$ and Lebesgue surface measure
  $d\Sigma_S$ and Leray measure
  $d\sigma_S(\xi)=|\nabla P(\xi)|^{-1}d\Sigma_S(\xi)$. Assume that
  \begin{align}
    \sup_{x\in X^d} (1+|x|)^r |(d\sigma_S)^\vee(x)|<\infty
  \end{align}
  for some $r>0$. Let $1\leq q\leq1+r$ and define
  \begin{align}
    \label{eq:defsigmaqr}
    \sigma(q,r) :=
    \begin{cases}
      \frac{2(d-1-r)q}{d-q} & \quad \text{if}\ \frac{d}{d-r}\leq q\leq 1+r\,,\\
      \frac{2rq+}{2rq-d(q-1)} & \quad \text{if}\ 1\leq q<\frac{d}{d-r}
    \end{cases}\,.
  \end{align}
  Here, $2rq+$ means $2rq + \epsilon$ with $\epsilon>0$ arbitrarily
  small but fixed. Then for all $W_1,W_2\in L^{2q}(X^d)$, we have
  \begin{align}
    \label{eq:franksabingen}
    \|W_1 F_S^*F_S W_2\|_{\cs^{\sigma(q,r)}}
    \lesssim \|W_1\|_{L^{2q}(X^d)}\|W_2\|_{L^{2q}(X^d)}\,,
  \end{align}
  where the implicit constant is independent of $W_1,W_2$.
\end{proposition}

\begin{remarks}
  \label{remfranksabingen}
  (1) We have $\sigma(q,r)\geq q$ when $r\leq(d-1)/2$
  and $\sigma(q,(d-1)/2)=\sigma(q)$ with $\sigma(q)$
  as in \eqref{eq:defsigmaq}.
    
  (2) The estimates \eqref{eq:tsschatten} and \eqref{eq:franksabingen}
  were proved for $\R^d$, but their (Fourier-analytic) proofs readily generalize
  to $\Z^d$.
    
  (3) The estimate in \cite[Proposition~A.5]{Cuenin2019} involved the resolvent
  of $P(-i\nabla)$. As usual, this implies \eqref{eq:franksabingen} since the
  imaginary part of the limiting resolvent equals the spectral measure.

  (4) Littman's bound $|(d\sigma_S)^\vee(x)|\lesssim\langle x\rangle^{-r}$ is rarely
  optimal except when the surface is completely flat in the vanishing curvature
  direction.
  (For a more detailed discussion and references to generic results, see, e.g.,
  \cite{CueninSchippa2022} by Schippa and the first author.)
  
  (5) As is discussed, e.g., in Ikromov, Kempe, and M\"uller
  \cite{Ikromovetal2010,IkromovMuller2011,IkromovMuller2016}, sharp decay
  estimates do not always imply $L^2\to L^p$ Fourier restriction bounds with
  optimal $p$.
\end{remarks}

\section{Bounds on number and sums of functions of eigenvalues in $L^2(\R^d)$}
\label{s:classiclt}

Suppose that the kinetic energy $T(-i\nabla)$ satisfies (1)-(3) in Assumption
\ref{assumt}.
Let $-e_1\leq -e_2\leq\cdots<0$ denote the negative eigenvalues
of $H=T-V$ in non-decreasing order (counting multiplicities) and
\begin{align}
  N_e(V) := \sum_{e_j>e}1 
\end{align}
denote the number of negative eigenvalues of $H$ below $-e\leq0$.
Let
\begin{align}
  \label{eq:defbs}
  BS(e) := |V|^{\frac12}(T+e)^{-1}V^{\frac12}
  \quad \text{on}\ L^2(\R^d)
\end{align}
where $V^{1/2}(x):=|V(x)|^{1/2}\sgn(V(x))$ with $\sgn(V(x))=1$ if $V(x)=0$.
By the Birman--Schwinger principle (cf.~\cite[Proposition 7.2]{Simon2005},
\cite[p. 99]{Cwikel1977}, \cite[Proposition 6]{Lieb1980T}), one has
\begin{align}
  \label{eq:boundnumber}
  N_e(V) = n(1,BS(e))
  \leq \|BS(e)\|_{\cs^{m,\infty}}^m \leq \|BS(e)\|_{\cs^m}^m
  \quad \text{for all}\ m>0\,.
\end{align}
As a consequence of the variational principle, i.e.,
\begin{align}
  \label{eq:noose}
  N_e(V) \leq N_e(V_+)
  = N_{e/2}(V_+-e/2) \leq N_{e/2}((V_+-e/2)_+)\,,
\end{align}
one can estimate for any $\gamma>0$,
\begin{align}
  \label{eq:lt}
  \tr(T(-i\nabla)-V)_-^\gamma
  \leq \gamma\int_0^\infty e^{\gamma-1}N_{e/2}((V_+-e/2)_+)\,de\,.
\end{align}

\subsection{Number of eigenvalues below a threshold}
We first prove estimates for Schatten norms of $BS(e)$.

\begin{theorem}
  \label{tracerestriction}
  Let $e>0$ and suppose $T(\xi)$ satisfies (1)-(3) in Assumption \ref{assumt}.
  \begin{enumerate}
  \item Let $m>d/s$. If $V\in L^m(\R^d)$, then there exists a constant
    $c_S>0$ (which also depends on $d,s,m,\tau$) such that
    \begin{align}
      \label{eq:tracerestriction1}
      \begin{split}
        \|BS(e)\|_m^m
        & \leq c_S(e^{1-m}\theta(1-e)+e^{d/s-m}\theta(e-1))\|V\|_m^m\,.
      \end{split}
    \end{align}

  \item \label{tracerestriction2}
    Suppose $T$ also satisfies (4) in Assumption \ref{assumt} with $r>0$.
    Let $q\in[1,r+1]$ and $m=\sigma(q,r)$ be as in
    \eqref{eq:defsigmaqr}.
    Suppose additionally $m>d/s$ and let $V\in L^m\cap L^q(\R^d)$.
    Then there is a constant $c_S$ (which also depends on $d,s,m,\tau,q,r$)
    such that
    \begin{align}
      \label{eq:tracerestriction2}
      \begin{split}
        \|BS(e)\|_m^m
        &  \leq c_S\left[\|V\|_m^m + \log(2+1/e)^m\|V\|_{q}^m\right] \theta(1-e)\\
        & \quad + c_S\left[e^{d/s-m}\|V\|_m^m + e^{-m}\min\{\|V\|_m,\|V\|_{q}\}^m\right] \theta(e-1)\,.
      \end{split}
    \end{align}

  \item In addition to the assumptions in (\ref{tracerestriction2}),
    suppose $q>d/s$. Then
    \begin{align}
      \label{eq:tracerestriction3}
      \begin{split}
        \|BS(e)\|_m^m
        \leq c_S\|V\|_{q}^m\left[\log(2+\frac1e)^m \theta(1-e)
          + e^{\frac{md}{sq}-m} \theta(e-1)\right]\,.
      \end{split}
    \end{align}
    If $q=d/s$, then \eqref{eq:tracerestriction3} holds with $\|\cdot\|_{m}^m$
    on the left side replaced by $\|\cdot\|_{m,\infty}^m$.
  \end{enumerate}
\end{theorem}

\begin{proof}
  We begin with the proof of \eqref{eq:tracerestriction1}.
  H\"older's inequality yields
  \begin{align}
    \label{eq:tomassteinschatten}
    \begin{split}
      \|W_1 F_{S}^*F_{S}W_2\|_{\cs^{1}(L^2(\R^d))}
      & \leq \|W_1F_{S}^*\|_{\cs^2(L^2(S,d\sigma_S),L^2(\R^d))}\|W_2F_{S}^*\|_{\cs^2(L^2(S,d\sigma_S),L^2(\R^d))}\\
      & = \|W_1\|_{L^{2}}\|W_2\|_{L^{2}}\sigma_S(S)\,.
    \end{split}
  \end{align}
  For $\tau>0$ as in Assumption \ref{assumt} we separate high
  and low energies using a bump function $\chi\in C_c^\infty(\R_+:[0,1])$ with
  $\supp\chi\subseteq[0,1]$.
  By the Kato--Seiler--Simon inequality \cite[Theorem 4.1]{Simon2005}
  (with $m>d/s$), we obtain
  \begin{align}
    \label{eq:highlow}
    \begin{split}
      & \|BS(e)\|_m^m\\
      & \quad \lesssim_m (\||V|^{\frac12}(T+e)^{-1}\chi(T/\tau)|V|^{\frac12}\|_m + \||V|^{\frac12}(T+e)^{-1}(1-\chi(T/\tau))|V|^{1/2}\|_m)^m\\
      & \quad \lesssim \||V|^{\frac12}(T+e)^{-1}\chi(T/\tau)|V|^{\frac12}\|_m^m + \|V\|_m^m\min\{1,e^{d/s-m}\}\,.
    \end{split}
  \end{align}
  To treat the low energy part we use the Lieb--Thirring trace inequality
  \cite[Theorem~9]{LiebThirring1976}
  \begin{align}
    \label{eq:lttrace}
    \|B^{1/2}AB^{1/2}\|_{\cs^m}^m \leq \|B^{m/2}A^m B^{m/2}\|_{\cs^1}\,, \quad m\geq1
  \end{align}
  for linear operators $A,B\geq0$ in a separable Hilbert space, the spectral theorem,
  and \eqref{eq:tomassteinschatten}. We obtain
  \begin{align}
    \label{eq:highlowaux}
    \begin{split}
      & \||V|^{\frac12}(T+e)^{-1}\chi(T/\tau)|V|^{\frac12}\|_m^m
      \leq \||V|^{\frac{m}{2}}(T+e)^{-m}\chi(T/\tau)^m |V|^{\frac{m}{2}}\|_{1}\\
      & \quad \leq \int_0^\tau dt\ \frac{\||V|^{m/2}F_{S_t}^*F_{S_t}|V|^{m/2}\|_{1}}{(t+e)^{m}}
      \leq \|V\|_{m}^m \int_0^\tau dt\ (t+e)^{-m} \sigma_{S_t}(S_t)\\
      & \quad \leq c_{S,\tau} \|V\|_{m}^m \int_0^\tau \frac{dt}{(t+e)^{m}}
      \leq c_{S,\tau,m} \|V\|_{m}^m \min\{e^{1-m},e^{-m}\}\,,
    \end{split}
  \end{align}
  where we used Assumption \ref{assumt} to estimate
  \begin{align}
    \label{eq:dependenceS}
    \sigma_{S_t}(S_t) \leq \sup_{t\in[0,\tau]}\sup_{\xi\in S_t}\frac{\Sigma_{S_t}(S_t)}{|\nabla P(\xi)|}
    \leq c_{S,\tau}\,.
  \end{align}
  Combining \eqref{eq:highlow} and \eqref{eq:highlowaux} proves
  \eqref{eq:tracerestriction1}.
  To prove \eqref{eq:tracerestriction2}, we proceed as in the proof of
  \eqref{eq:tracerestriction1} but estimate the low energies using
  the Stein--Tomas estimate for trace ideals \eqref{eq:franksabingen}
  instead.
  For $q\in[1,r+1]$, we obtain
  \begin{align}
    \label{eq:tstrace}
    \begin{split}
      \||V|^{\frac12}(T+e)^{-1}\chi(T/\tau)|V|^{\frac12}\|_{\sigma(q,r)}
      & \leq \int_0^\tau \frac{dt}{t+e} \||V|^{1/2}F_{S_t}^*F_{S_t}|V|^{1/2}\|_{\sigma(q,r)}\\
      & \leq c_{S} \min\{\log(1+\tau/e),\tau/e\} \|V\|_{q}\,.
    \end{split}
  \end{align}
  Setting $m=\sigma(q,r)$ on the left side of \eqref{eq:tstrace} and combining
  it with \eqref{eq:tracerestriction1} yields \eqref{eq:tracerestriction2}.

  The final estimate \eqref{eq:tracerestriction3} follows from the
  proof of \eqref{eq:tracerestriction2} by replacing the estimate for
  the high energies in the second and third line of \eqref{eq:highlow}
  by the following estimate,
  \begin{align*}
    \||V|^{\frac12}(T+e)^{-1}(1-\chi(T/\tau))|V|^{1/2}\|_m^m
    & \leq \||V|^{\frac12}(T+e)^{-1}(1-\chi(T/\tau))|V|^{1/2}\|_{q}^m\\
    & \lesssim \|V\|_{q}^m\min\{1,e^{\frac{md}{sq}-m}\}\,,
  \end{align*}
  where $m\geq q>d/s$. (Here we used the Kato--Seiler--Simon inequality
  again.) This concludes the proof of \eqref{eq:tracerestriction3}.
  For $q=d/s$ we use Cwikel's bound (see \cite[Theorem 4.2]{Simon2005}
  or \eqref{eq:cwikel} below), which is applicable since $q>1$ in this case.
\end{proof}

\begin{remarks}
  \label{jtheigenvalue}
  (1) The terms proportional to $\|V\|_m^m$ in \eqref{eq:tracerestriction2}
  and the term that scales like $e^{\frac{md}{sq}-m}$ in
  \eqref{eq:tracerestriction3} are due to high energies.
  
  (2) If $r=(d-1)/2$, $d>s\geq d/q$ and $0<e<1$,
  then \eqref{eq:tracerestriction3}
  implies for $q\in(1,(d+1)/2]$ and $m=\sigma(q)$,
  \begin{align}
    N_e(V) \leq c_{S} \log(1/e)^{\sigma(q)} \|V\|_{q}^{\sigma(q)}\,.
  \end{align}
  Thus, the $n$-th negative eigenvalue $-1<-e_n<0$ satisfies
  \begin{align}
    \label{eq:bounden2}
    e_n \leq \exp\left(-\frac{n^{1/\sigma(q)}}{c_{S}^{1/\sigma(q)}\,\|V\|_{q}}\right)\,.
  \end{align}
\end{remarks}

We close this subsection by proving a slight refinement of the bound for $N_e(V)$ that follows immediately from \eqref{eq:boundnumber} and \eqref{eq:tracerestriction1}. To that end we apply Fan's inequality \cite{Fan1951} (see also \cite[Theorem~1.7]{Simon2005}), which asserts
\begin{align}
  \label{eq:fan}
  s_{j+\ell+1}(A+B) \leq s_{j+1}(A) + s_{\ell+1}(B)
\end{align}
for all $j,\ell\in\N_0$ and all $A,B\in\cs^\infty$.

\begin{corollary}
  \label{triangle}
  Let $e,\tau>0$, $m_1,m_2\geq1$.
  Let $L_{\rm loc}^1(\R^d)\ni T(\xi)\geq0$ and $V\in L_{\rm loc}^1(\R^d)$ so that
  $BS_<(e):=|V|^{\frac12}(T+e)^{-1}\one_{\{T<\tau\}}V^{\frac12}$ and
  $BS_>(e):=|V|^{\frac12}(T+e)^{-1}\one_{\{T>\tau\}}V^{\frac12}$ are compact operators.
  Then
  \begin{align}
    \label{eq:triangle1}
    N_e(V) \leq 2\cdot[2^{m_1}\|BS_<(e)\|_{\cs^{m_1,\infty}}^{m_1} + 2^{m_2}\|BS_>(e)\|_{\cs^{m_2,\infty}}^{m_2}]\,.
  \end{align}
  In particular, for $T(\xi)$ satisfying (1)-(3) in Assumption \ref{assumt},
  $m_1>1$ and $m_2>d/s$, there is $c_S>0$ (which also depends on $d,s,m_1,m_2,\tau$)
  such that
  \begin{align}
    \label{eq:triangle2}
    N_e(V)
    \leq c_S [(e^{1-m_1}\|V\|_{m_1}^{m_1}+\|V\|_{m_2}^{m_2})\theta(1-e) + e^{d/s-m_2}\|V\|_{m_2}^{m_2}\theta(e-1)]\,.
  \end{align}
\end{corollary}

\begin{proof}
  We begin with proving \eqref{eq:triangle1}.
  By \eqref{eq:fan}, we have
  $s_{n+1}(BS(e))\leq s_{n/2+1}(BS_<(e))+s_{n/2+1}(BS_>(e))$ for even $n$ and
  $s_{n+1}(BS(e))\leq s_{(n+1)/2+1}(BS_<(e))+s_{(n-1)/2+1}(BS_>(e))$ for odd $n$.
  Thus,
  \begin{align}
    \begin{split}
      & \{n\in2\N_0:\, s_{n+1}(BS(e))>1\}\\
      & \quad \subseteq \{n\in2\N_0:\,s_{n/2+1}(BS_<(e))>1/2\} \cup \{n\in2\N_0:\,s_{n/2+1}(BS_>(e))>1/2\}
    \end{split}
  \end{align}
  and a similar statement holds for odd $n$. Combining this with
  \begin{align*}
    & \{n\in\N_0:\, s_{n+1}(BS(e))>1\}\\
    & \quad = \{n\in 2\N_0:\, s_{n+1}(BS(e))>1\} \cup \{n\in(2\N_0+1):\, s_{n+1}(BS(e))>1\}
  \end{align*}
  and \eqref{eq:boundnumber} yields \eqref{eq:triangle1}, because
  \begin{align}
    \begin{split}
      N_e(V) & = \# \{n\in\N:\, s_{n}(BS(e))>1\}\\
      & \leq 2\left(\#\{n\in \N_0:\,s_{n+1}(BS_<(e))>1/2\} + \# \{n\in\N_0:\,s_{n+1}(BS_>(e))>1/2\}\right)\\
      & \leq 2 \left( 2^{m_1}\|BS_<(e)\|_{\cs^{m_1,\infty}}^{m_1} + 2^{m_2}\|BS_<(e)\|_{\cs^{m_2,\infty}}^{m_2} \right)\,.
    \end{split}
  \end{align}
  To prove \eqref{eq:triangle2}, we first write
  \begin{align}
    \label{eq:triangleaux1}
    N_e(V) = n(1,BS(e))\theta(1-e) + n(1,BS(e))\theta(e-1)\,.
  \end{align}
  The second summand is estimated using the Kato--Seiler--Simon inequality by
  \begin{align}
    n(1,BS(e))\theta(e-1)
    \leq \|BS(e)\|_{m_2}^{m_2}\theta(e-1)
    \lesssim \|V\|_{m_2}^{m_2} \cdot e^{d/s-m_2}\theta(e-1)\,.
  \end{align}
  The first summand in \eqref{eq:triangleaux1} is estimated using
  \eqref{eq:triangle1} by
  \begin{align}
    \begin{split}
      n(1,BS(e))\theta(1-e)
      & \lesssim (\|BS_<(e)\|_{m_1}^{m_1} + \|BS_>(e)\|_{m_2}^{m_2})\theta(1-e)\\
      & \lesssim (e^{1-m_1}\|V\|_{m_1}^{m_1} + \|V\|_{m_2}^{m_2})\theta(1-e)\,,
    \end{split}
  \end{align}
  where we used the steps in the proof of \eqref{eq:tracerestriction1}.
\end{proof}

\subsection{Sums of powers of eigenvalues}

We now use \eqref{eq:boundnumber}-\eqref{eq:lt} and Theorem
\ref{tracerestriction} to obtain estimates for sums of powers of
eigenvalues of $T-V$.

\begin{theorem}
  \label{classiclt}
  Suppose $T(\xi)$ satisfies (1)-(3) in Assumption \ref{assumt}.
  \begin{enumerate}
  \item If $\gamma>0$ and
    $V\in L^{\gamma+1}\cap L^{\gamma+d/s}(\R^d)$ then there exists a constant
    $c_S>0$ (which also depends on $d,s,m,\gamma$) such that
    \begin{align}
      \label{eq:classiclt}
      \tr_{L^2(\R^d)}(T(-i\nabla)-V)_-^\gamma
      \leq c_S \int_{\R^d} (V_+(x)^{\gamma+1} + V_+(x)^{\gamma+d/s})\,dx\,.
    \end{align}

  \item Suppose $T$ also satisfies (4) in Assumption \ref{assumt} with $r>0$.
    Let $q\in[1,r+1]$ and $m=\sigma(q,r)$.
    Suppose additionally $m>d/s$.
    If $\gamma>m-d/s$ and $V\in L^q\cap L^{\gamma+d/s}(\R^d)$,
    then there is a constant $c_S$ (which also depends on
    $d,s,m,q,\gamma,r$) such that
    \begin{align}
      \label{eq:classiclt2}
      \tr_{L^2(\R^d)}(T(-i\nabla)-V)_-^\gamma
      \leq c_S (\|V_+\|_q^m + \|V_+\|_{\gamma+d/s}^{\gamma+d/s})\,.
    \end{align}
  \end{enumerate}
\end{theorem}

\begin{proof}
  By the variational principle we can assume $V=V_+\geq0$.
  To prove \eqref{eq:classiclt} we apply \eqref{eq:triangle2}
  in Corollary \ref{triangle} for any $m_1>1$ and $m_2>d/s$
  and obtain
  \begin{align*}
    N_{e/2}((V(x)-e/2)_+)
    \leq c_{S}\left[e^{1-m_1}\int_{\R^d}(V(x)-e/2)_+^{m_1} + e^{d/s-m_2} \int_{\R^d}(V(x)-e/2)_+^{m_2}\right]\,.
  \end{align*}
  Plugging this into \eqref{eq:lt} with $\gamma>\max\{m_1-1,m_2-d/s\}$ yields
  \begin{align*}
    & \tr(T(-i\nabla)-V)_-^\gamma\\
    & \quad \leq c_{S} \int_0^\infty de\, \left[e^{\gamma-m_1}\int_{\R^d}dx\ (V(x)-e/2)_+^{m_1} + e^{\gamma-1+d/s-m_2}\int_{\R^d}dx\ (V(x)-e/2)_+^{m_2}\right] \\
    & \quad \leq c_S \int_{\R^d}dx\, (V(x)^{\gamma+1} + V(x)^{\gamma+d/s})\,,
  \end{align*}
  where $c_S$ also depends on $d,s,m,\gamma$.
  This proves \eqref{eq:classiclt}.
  
  To prove \eqref{eq:classiclt2} we instead use \eqref{eq:tracerestriction2}
  in Theorem \ref{tracerestriction} and plug the right side of
  \begin{align*}
    & N_{e/2}((V(x)-e/2)_+)\\
    & \quad \leq c_S (e^{d/s-m}\theta(e-1)+\theta(1-e))\int_{\R^d}(V(x)-e/2)_+^{m}\\
    & \qquad\quad + c_S \theta(1-e)\log(2+1/e)^m\left(\int_{\R^d}(V(x)-e/2)_+^{q}\right)^{\frac{m}{q}}\\
    & \quad \leq c_S e^{d/s-m}\int_{\R^d}(V(x)-e/2)_+^{m} + c_S \theta(1-e)\log(2+1/e)^m\left(\int_{\R^d}(V(x)-e/2)_+^{q}\right)^{\frac{m}{q}}
  \end{align*}
  into \eqref{eq:lt}.
  For $\gamma+d/s>m>d/s$ the first summand gives again rise to
  \begin{align*}
    \int_0^\infty de\ e^{\gamma-1+d/s-m}\int_{\R^d}(V(x)-e/2)_+^m\,dx
    \leq c_{S}\int_{\R^d} V(x)^{\gamma+d/s}\,dx\,,
  \end{align*}  
  whereas the second summand contributes with
  \begin{align*}
    \int_0^1 de\ e^{\gamma-1}\log(2+1/e)^m \left(\int_{\R^d}(V(x)-e/2)_+^{q}\right)^{\frac{m}{q}}
    \leq c_{S} \|V\|_q^m
  \end{align*}
  to the left hand side of \eqref{eq:classiclt2} for all $\gamma>0$.
  (As before, $c_S$ also depends on $d,s,m,q,\gamma,r$)
  This concludes the proof.
\end{proof}

\begin{remark}
  (1) The term $\|V_+\|_{\gamma+d/s}^{\gamma+d/s}$ on the right sides of
  \eqref{eq:classiclt}-\eqref{eq:classiclt2} comes from high
  energies as can be seen from the proofs of
  \eqref{eq:tracerestriction1}-\eqref{eq:tracerestriction2}.
  In Theorem \ref{classicltlattice} we will see that this term is absent
  for operators in $\ell^2(\Z^d)$ since only low energies are present.

  (2) The term $\|V_+\|_{\gamma+d/s}^{\gamma+d/s}$ is necessary, which can
  be seen by repeating the
  arguments in the proof of Theorem \ref{tracerestriction} with
  $|\Delta+\mu|$ and letting $\mu\to0$.
  
  (3) Estimate \eqref{eq:classiclt} also holds in case the level sets
  of $T$ are not curved and can be seen as a Lieb--Thirring inequality
  since the right hand side is ``local'' in the sense that it involves only
  integrals over $V$.
  In contrast, \eqref{eq:classiclt2} requires non-vanishing Gaussian curvature
  of the level sets. Moreover, \eqref{eq:classiclt2} is non-local in the
  sense that it involves powers of integrals of $V$.
\end{remark}

\subsection{Sums of logarithms of eigenvalues}
\label{s:loglt}

Suppose $r=(d-1)/2$, let $s\in[2d/(d+1),d)$, $V\in L^{(d+1)/2}$,
and assume that $T$ satisfies (1)-(4) in Assumption \ref{assumt}.
By \cite[Theorem 4.2]{CueninMerz2021}, for any eigenvalue $a_j^S>0$
of the operator $\cv_S=F_SVF_S^*$
in $L^2(S)$, there exists a negative eigenvalue $-e_j(\lambda)<0$ of
$H_\lambda=T-\lambda V$ with weak coupling limit
\begin{align}
  \label{eq:weakcoupling}
  e_j(\lambda) = \exp\left(-\frac{1}{2\lambda a_j^S}(1 + o(1))\right)\,,
  \quad \lambda\to0\,.
\end{align}
(Eigenvalues $-e_j<0$ corresponding to zero-eigenvalues of $\cv_S$ obey
$e_j(\lambda)=\me{-c_j\lambda^{-2}}$ for some $c_j>0$,
cf.~\cite[Theorem 4.4]{CueninMerz2021}).
On the other hand, as we have seen in \eqref{eq:bounden2} in Remark
\ref{jtheigenvalue}, if the $j$-th eigenvalue $-e_j(\lambda)$ is
greater than $-1$, then it satisfies
\begin{align}
  \label{eq:boundenoncemore}
  e_j(\lambda) \leq \exp\left(-\frac{j^{1/\sigma(q)}}{c_{S}^{1/\sigma(q)}\lambda \|V\|_{q}}\right)\,,
  \quad q\in\left[\frac{d}{s},\frac{d+1}{2}\right]\,.
\end{align}

Formulae \eqref{eq:weakcoupling}-\eqref{eq:boundenoncemore} illustrate
that the eigenvalues of $H_\lambda$ approach
$\inf\sigma_{\mathrm{ess}}(H_\lambda)=0$ exponentially fast. This suggests
to compute logarithmic moments of eigenvalues,
\begin{align*}
  \sum_{j} \left(\frac{1}{\log(\langle1/e_j\rangle)}\right)^\gamma\,, \quad \gamma>0\,.
\end{align*}
(Note that $1/\log(1/x)\geq x$ for $0<x<1/2$, say.)
For those eigenvalues $e_j(\lambda)$ corresponding to the $a_j^S$ in
the asymptotics \eqref{eq:weakcoupling}, estimate \eqref{eq:tsschatten}
implies, for $\lambda$ in a sufficiently small open neighborhood of $0$,
\begin{align}
  \label{eq:logltprelim}
  \begin{split}
    \sum_{j} \left(\frac{1}{\log(\langle1/e_j(\lambda)\rangle)}\right)^{\sigma(q)}
    & \sim \sum_j \left(1+\frac{1}{2\lambda a_S^j}\right)^{-\sigma(q)}
    \sim \lambda^{\sigma(q)} \tr (\cv_S)_+^{\sigma(q)}\\
    & \lesssim \lambda^{\sigma(q)}\|V\|_{L^q}^{\sigma(q)}\,,
  \end{split}
\end{align}
where $\sigma(q)$ is as in \eqref{eq:defsigmaq} and $q\in[1,(d+1)/2]$.
We now prove analogous estimates for $\lambda=1$, in which case we
cannot use the results in the weak coupling regime.

\begin{theorem}
  \label{loglt}
  Let $H=T-V$ with $T$ satisfying (1)-(4) in Assumption \ref{assumt} with $r>0$.
  Let $q\in[1,r+1]$ and $m=\sigma(q,r)$.
  Suppose additionally $m>d/s$ and let $V\in L^m\cap L^q(\R^d)$.
  Then for any $\gamma>m$ there is a constant $c_S$ (which also depends on
  $d,s,m,q,\gamma,r$) such that
  \begin{align}
    \label{eq:logltfinal}
    \sum_j [\log(\langle 1/e_j\rangle)]^{-\gamma}
    \leq c_{S} \|V_+\|_m^m + \|V_+\|_q^{m}\,.
  \end{align}
  Moreover, if $q\geq d/s$, then
  \begin{align}
    \label{eq:logltfinal2}
    \sum_j [\log(\langle 1/e_j\rangle)]^{-\gamma}
    \leq c_S \|V_+\|_q^{m}\,.
  \end{align}
\end{theorem}

\begin{proof}
  By the variational principle we can again assume $V=V_+$.
  To estimate the left side of \eqref{eq:logltfinal} we use
  \begin{align}
    \label{eq:logltaux}
    \begin{split}
      \frac{1}{\left(\log(\langle e^{-1}\rangle)\right)^\gamma}
      = \gamma\int_0^{e} \left(\log(\langle r^{-1}\rangle)\right)^{-\gamma-1}\cdot\left\langle\frac{1}{r}\right\rangle^{-2}\frac{dr}{r^3}
    \end{split}
  \end{align}
  for $\gamma>0$. Thus,
  \begin{align}
    \label{eq:logltaux2}
    \begin{split}
      \sum_j \frac{1}{\left(\log(\langle e_j^{-1}\rangle)\right)^\gamma}
      & = \gamma\int_0^\infty \left(\log(\langle r^{-1}\rangle)\right)^{-\gamma-1}\cdot\left\langle\frac{1}{r}\right\rangle^{-2}\sum_j\theta(e_j-r)\frac{dr}{r^3}\\
      & = \gamma\int_0^\infty \left(\log(\langle r^{-1}\rangle)\right)^{-\gamma-1}\cdot\left\langle\frac{1}{r}\right\rangle^{-2}N_{r}(V)\frac{dr}{r^3}\,.
    \end{split}
  \end{align}
  By \eqref{eq:boundnumber} and \eqref{eq:tracerestriction2} in Theorem
  \ref{tracerestriction} for $m>d/s$,
  we estimate
  \begin{align}
    \label{eq:logltaux3}
    N_{r}(V) \lesssim \|V\|_m^m + \left(\log\left(2+\frac{1}{r}\right)\right)^m \|V\|_q^m\,.
  \end{align}
  Thus, the left side of \eqref{eq:logltfinal} can be estimated by
  \begin{align}
    \label{eq:loglt1}
    \begin{split}
      \sum_j \left(\frac{1}{\log(\langle 1/e_j\rangle)}\right)^\gamma
      & \lesssim \|V\|_m^m \int_0^\infty \frac{dr}{r^3}\langle r^{-1}\rangle^{-2}\left(\log(\langle r^{-1}\rangle)\right)^{-\gamma-1}\\
      & \quad + \|V\|_q^m \int_0^\infty \frac{dr}{r^3}\langle r^{-1}\rangle^{-2}\left(\log(\langle r^{-1}\rangle)\right)^{-\gamma-1}\cdot \left(\log\left(2+\frac{1}{r}\right)\right)^{m}\\
      & \lesssim \|V\|_m^m + \|V\|_q^m\,.
    \end{split}
  \end{align}
  This concludes the proof of \eqref{eq:logltfinal}.
  The proof of \eqref{eq:logltfinal2} is completely analogous, but
  uses \eqref{eq:tracerestriction3} instead of \eqref{eq:tracerestriction2}.
  Thus, estimate \eqref{eq:logltaux3} is replaced by
  \begin{align}
    N_{r}(V) \lesssim \|V\|_m^q \left[1 + \left(\log\left(2+\frac{1}{r}\right)\right)^m \right]\,.
  \end{align}
  Proceeding as in the proof of \eqref{eq:logltfinal} concludes the proof of
  \eqref{eq:logltfinal2}.
\end{proof}

\begin{remarks}
  \label{logltremark}
  (1) In contrast to the right side of \eqref{eq:classiclt}, the powers of $V$
  appearing on the right side of \eqref{eq:logltfinal} are all the same.

  (2) For $r=(d-1)/2$ and $m=d+1$ the power $d+1$ on the right side of
  \eqref{eq:logltfinal} is consistent with that on the right side of
  \eqref{eq:logltprelim}. However, \eqref{eq:logltfinal} is slightly
  weaker than \eqref{eq:logltprelim} due to the assumption $\gamma>d+1$
  and, if $q<d/s$, the additional $\|V\|_{d+1}^{d+1}$ term on the right of
  \eqref{eq:logltfinal}.

  (3) We do not know whether the restriction $\gamma>m$ (especially
  $\gamma>d+1$ for $r=(d-1)/2$ and $m=d+1$) is necessary.
\end{remarks}

\subsection{CLR bounds in $L^2(\R^d)$}
\label{s:cwikel}

Recall that $N_0(V_+)$ equals the number of eigenvalues of
$V_+^{1/2}T^{-1}V_+^{1/2}$ above one, which can be estimated by
\eqref{eq:boundnumber}.
Formula \eqref{eq:evboundtraceideal}, Cwikel's bound
\cite{Cwikel1977}, i.e.,
\begin{align}
  \label{eq:cwikel}
  \|f(-i\nabla)g(x)\|_{\cs^{p,\infty}(L^2(\R^d))}
  \lesssim_p \|f\|_{L^{p,\infty}(\R^d)}\|g\|_{L^p(\R^d)}\,,
  \quad p\in(2,\infty)\,,
\end{align}
and \eqref{eq:boundnumber} yield the classical
Cwikel--Lieb--Rosenbljum (CLR) bound
\cite{Cwikel1977,Lieb1976B,Rosenbljum1972}
\begin{align*}
  \|V_+^{1/2}(-\Delta)^{-1}V_+^{1/2}\|_{\cs^{d/2,\infty}(L^2(\R^d))}
  \lesssim_d \||\xi|^{-1}\|_{L^{d,\infty}} \|V_+\|_{L^{d/2}}
\end{align*}
for the number of negative eigenvalues $N_{0}(V)$ of $T-V$ when
$T=-\Delta$ in $d\geq3$. Such bounds can never hold in $d=1,2$, or
for $T$ satisfying Assumption \ref{assumt} in $d\geq2$ due to the
existence of weakly coupled bound states
\cite{Simon1976,Laptevetal2002,Franketal2007T,HainzlSeiringer2008C,HainzlSeiringer2010,Hoangetal2022,CueninMerz2021}.

Interestingly, using \eqref{eq:cwikel}, one does obtain a CLR bound for powers
$T=|\Delta+1|^{1/\sigma}$ of the BCS operator $|\Delta+1|$ in $L^2(\R^2)$
when $\sigma>1$. This follows from
the uniform bound $\int_0^1 (t+e)^{-1/\sigma}\,dt\lesssim_s1$ for all $e\geq0$
and $\sigma>1$.
The following theorem generalizes this observation to
powers $T^{1/\sigma}$ with $T$ satisfying (1)-(3) in Assumption \ref{assumt}
to all $d\in\N$.
The proof is inspired by Frank \cite{Frank2014}, which, in turn, uses ideas
of Rumin \cite{Rumin2010,Rumin2011}.

\begin{theorem}
  \label{clrbcs}
  Let $d\in\N$, $\sigma>1$, and
  suppose $T(\xi)$ satisfies (1)-(3) in Assumption \ref{assumt} with
  the weaker assumption $s\leq d$.
  Then, for $V\in L^{\sigma}\cap L^{\frac{\sigma d}{s}}(\R^d)$, one has
  \begin{align}
    \label{eq:clrbcs}
    n(1,|V|^{\frac12}T^{-1/\sigma}V^{\frac12})
    \lesssim_{S,\sigma,s,d,\tau} \|V_+\|_{L^{\sigma}(\R^d)}^\sigma + \|V_+\|_{L^{\sigma d/s}(\R^d)}^{\sigma d/s}\,.
  \end{align}
\end{theorem}

\begin{proof}
  By the variational principle we can assume $V=V_+$.
  We first show how to prove \eqref{eq:clrbcs} for $s=d\in\N$ using
  Cwikel's estimate \eqref{eq:cwikel}. For $\beta>0$ a straightforward
  computation shows
  \begin{align}
    \label{eq:clrbcsaux1}
    \begin{split}
      \|T(\xi)^{-1/(2\sigma)}\|_{L^{2p,\infty}(\R^d)}^{2p}
      & = \sup_{\beta>0}\beta^{-2p} \left|\{\xi\in\R^d:\, T(\xi)^{-1/(2\sigma)}>1/\beta\}\right|\\
      & \lesssim \sup_{\beta>0} \beta^{-2p}\left(\beta^{2\sigma}\one_{\{\beta\leq1\}} + \beta^{2\sigma\cdot d/s}\one_{\{\beta\geq1\}}\right)\,.
    \end{split}
  \end{align}
  For the right side to be finite we need $p=\sigma$ and $s=d$.
  Thus, by \eqref{eq:boundnumber} and Cwikel's estimate \eqref{eq:cwikel},
  we obtain for $\sigma>1$,
  \begin{align*}
    n(1,|V|^{\frac12}T^{-1/\sigma}V^{\frac12})
    \leq \|V^{\frac12}T^{-\frac1\sigma}V^{\frac12}\|_{\cs^{\sigma,\infty}(L^2(\R^d))}^\sigma
    \leq \|T^{-\frac{1}{2\sigma}}V^{\frac12}\|_{\cs^{2\sigma,\infty}(L^2(\R^d))}^{2\sigma}
    \lesssim \|V\|_{L^{\sigma}(\R^d)}^\sigma\,,
  \end{align*}
  which concludes the proof for $s=d$. We will now show \eqref{eq:clrbcs}
  for $s<d$ by proceeding as in \cite{Frank2014}.
  Let $\Gamma$ be an arbitrary operator in $L^2(\R^d)$ satisfying
  $0\leq\Gamma\leq T^{-1/\sigma}$ and $\rho_\Gamma(x):=\Gamma(x,x)$.
  Let $P_E:=\one_{(E,\infty)}(T^{1/\sigma})$ and $P_E^\perp=1-P_E$.
  We shall now estimate
  \begin{align}
    \label{eq:clrbcsaux2}
    \tr(\Gamma^{1/2}T^{1/\sigma}\Gamma^{1/2})
    = \int_{\R^d} dx\int_0^\infty dE\, (P_E\Gamma P_E)(x,x)
  \end{align}
  from below.
  By a density argument it suffices to consider the case where $\Gamma$
  has finite-rank and smooth eigenfunctions.
  For any subset $\Omega\subseteq\R^d$ of finite measure we have
  \begin{align}
    \label{eq:clrbcsaux3}
    \begin{split}
      \left(\int_\Omega\rho_\Gamma(x)\,dx\right)^{1/2}
      & = \|\Gamma^{1/2}\one_\Omega\|_2
      \leq \|\Gamma^{1/2}P_E\one_\Omega\|_2 + \|\Gamma^{1/2}P_E^\perp\one_\Omega\|_2\\
      & \leq \|\Gamma^{1/2}P_E\one_\Omega\|_2 + |\Omega|^{1/2}\sqrt{F(E)}
    \end{split}
  \end{align}
  where we used $\Gamma\leq T^{-1/\sigma}$ and defined
  \begin{align}
    \label{eq:clrbcsaux4}
    \begin{split}
      F(E) & := \frac{\|T^{-1/(2\sigma)}P_E^\perp\one_\Omega\|_2^2}{|\Omega|}
      = \int_{\R^d}\frac{d\xi}{T(\xi)^{1/\sigma}}\one_{\{T(\xi)^{1/\sigma}<E\}}\\
      & = \int_{\R^d}d\xi\, \one_{\{T(\xi)^{1/\sigma}<E\}}\int_0^\infty dz\, \one_{\{z<T(\xi)^{-1/\sigma}\}} \\
      & = \int_0^\infty dz\, |\{\xi\in\R^d:\, T(\xi)^{1/\sigma}<E\wedge z^{-1}\}|\,. 
    \end{split}
  \end{align}
  By \eqref{eq:clrbcsaux1} we have
  \begin{align}
    \label{eq:clrbcsaux5}
    \begin{split}
      F(E)
      & \leq \int_0^\infty dz \left[(E\wedge z^{-1})^\sigma \one_{\{E\wedge z^{-1}<1\}} + (E\wedge z^{-1})^{\sigma d/s} \one_{\{E\wedge z^{-1}>1\}}\right]\\
      & = E^{\sigma-1}\left[\int_0^1 dz\, \one_{\{E<1\}} + \int_1^\infty dz\, z^{-\sigma}\one_{\{z>E\}}\right] \\
      & \quad + E^{\sigma d/s-1}\left[\int_0^1 dz\, \one_{\{E>1\}} + \int_1^\infty dz\, z^{-\sigma d/s}\one_{\{z<E\}}\right]\\
      & \sim E^{\sigma-1} + E^{\sigma d/s-1}\,.
    \end{split}
  \end{align}
  From this and \eqref{eq:clrbcsaux3}, it follows from Lebesgue's
  differentiation theorem that
  \begin{align}
    \label{eq:clrbcsaux6}
    \begin{split}
      (P_E \Gamma P_E)(x,x)
      \geq \left(\sqrt{\rho_\Gamma(x)} - \sqrt{F(E)}\right)_+^2
      \geq \left(\sqrt{\rho_\Gamma(x)} - c\cdot(E^{\frac{\sigma-1}{2}} + E^{\frac{\sigma d/s-1}{2}})\right)_+^2
    \end{split}
  \end{align}
  for almost every $x\in\R^d$.
  Integration over $E$ shows
  \begin{align}
    \label{eq:clrbcsaux7}
    \tr(\Gamma T^{1/\sigma})
    \gtrsim \int_{\R^d}dx\, \left(\rho_\Gamma(x)^{\frac{\sigma}{\sigma-1}}\one_{\{\rho_\Gamma\leq1\}} + \rho_\Gamma(x)^{\frac{\sigma d/s}{\sigma d/s-1}}\one_{\{\rho_\Gamma\geq1\}}\right) 
  \end{align}
  for all $0\leq\Gamma\leq T^{-1/\sigma}$. By a slight generalization of
  the duality principle in \cite[Lemma~2.4]{Frank2014}, formula
  \eqref{eq:clrbcsaux7} is equivalent to
  \begin{align}
    \tr(T^{-\frac{1}{2\sigma}}VT^{-\frac{1}{2\sigma}}-\mu)_+
    \lesssim \mu^{-\sigma+1}\int_{\R^d}V^{\sigma}(x)\,dx + \mu^{-\frac{\sigma d}{s}+1}\int_{\R^d}V^{\frac{\sigma d}{s}}(x)\,dx
  \end{align}
  for any $\mu>0$. Noting that for any $\mu<1$, we have
  \begin{align}
    n(1,V^{\frac12}T^{-1/\sigma}V^{\frac12})
    = n(1,T^{-\frac{1}{2\sigma}}VT^{-\frac{1}{2\sigma}})
    \leq (1-\mu)^{-1}\tr(T^{-\frac{1}{2\sigma}}VT^{-\frac{1}{2\sigma}}-\mu)_+\,,
  \end{align}
  which concludes the proof of Theorem \ref{clrbcs}.
\end{proof}

\section{Schr\"odinger operators with degenerate kinetic energy in $\ell^2(\Z^d)$}
\label{s:discretesetting}

In this section we prove analogs of the previous results for lattice 
Schr\"odinger operators. 
We first review two instances of the Laplacian in $\ell^2(\Z^d)$
and discuss an analog of the BCS operator $|\Delta+1|$ in $\ell^2(\Z^d)$.
Subsequently, we state and prove our results on numbers and sums of
functions of eigenvalues.

\subsection{Laplace and BCS-type operators in $\ell^2(\Z^d)$}

\subsubsection{Ordinary lattice Laplace}

The standard lattice Laplacian is defined by 
\begin{align}
  \label{eq:deflatticelaplacepositionspace}
  - \Delta u(n)
 = \frac{1}{2d}\sum_{\|m-n\|_2=1}u(m).
\end{align}
Its spectrum is absolutely continuous and equal to $[-1,1]$.
The Fourier multiplier associated to
\eqref{eq:deflatticelaplacepositionspace} is given by
$d^{-1}\sum_{j=1}^d\cos(2\pi\xi_j)$.
Let $Z$ denote the set of critical values of this symbol.
The level sets
\begin{align}
  \label{eq:fermisurfacelatticelaplace}
  S_t := \left\{\xi\in\T^d:\,\frac{1}{d}\sum_{j=1}^d\cos(2\pi\xi_j)=t\right\}\,,
  \quad t\in [-1,1]\setminus Z
\end{align}
are strictly convex (i.e.\ have everywhere positive Gaussian curvature) in $d=2$,
cf.~\cite[Lemma 3.3]{Schlagetal2002}.
This implies that, for $d=2$, item (4) in Assumption \ref{assumt} holds with $r=1/2$.
In higher dimensions, $S_t$ is not convex for $|t|<1-2/d$,
cf.~\cite{ShabanVainberg2001}.
For $d=3$, Erd\H{o}s--Salmhofer \cite{ErdosSalmhofer2007} obtained the sharp decay
of the Fourier transform of the surface measure up to logarithmic factors.
Recently, Schippa and the first author \cite{CueninSchippa2022} provided a
simpler proof and obtained the sharp bound
\begin{align*}
  |(d\Sigma_{S_t})^\vee(x)|\lesssim (1+|x|)^{-3/4},
\end{align*}
provided $t\neq 1$ (The level set $S_1$ contains a flat umbilic point, see
\cite{ErdosSalmhofer2007}).
This also follows from results of Taira \cite{Taira2021U}, which are based
on the work of Ikromov and M\"uller \cite{IkromovMuller2016} and involve
Newton polygon methods. We conclude that, for $d=3$, assumption (4) holds
with $r=3/4$. We are not aware of any sharp estimates in dimensions $d>3$.

\subsubsection{Molchanov--Vainberg Laplace}
\label{ss:molchanovvainberg}

Molchanov and Vainberg \cite{MolchanovVainberg1999} considered the following
modification of $-\Delta$, which is defined by
\begin{align}
  \label{eq:defmolvainlaplaceposition}
  -\Delta_{\MV}\psi(n) = 2^{-d}\sum_{\|m-n\|_2=\sqrt d}\psi(m)\,.
\end{align}
Again, its spectrum is absolutely continuous and equal to $[-1,1]$.
The level sets of the associated Fourier multiplier
$\prod_{j=1}^d\cos(2\pi\xi_j)$
are given by
\begin{align*}
  S_t:= \left\{\xi\in\T^d:\,\prod_{j=1}^d\cos(2\pi\xi_j)=t\right\}\,,
  \quad t\in [-1,1]\setminus Z\,.
\end{align*}
The advantage over the standard Laplacian is that the level sets $S_t$
are strictly convex for all $t\in(-1,1)$ as Poulin
\cite[Theorems~1.1 and 3.4]{Poulin2007} showed. Hence, for the
Molchanov--Vainberg Laplacian, item~(4) in Assumption \ref{assumt} holds
with $r=(d-1)/2$ for all $d\geq2$.

\subsubsection{An analog of the BCS operator in $\ell^2(\Z^d)$}

We can define the analog of the BCS operator in $\ell^2(\Z^d)$
by the Fourier multiplier
\begin{align}
  \label{eq:deflatticebcs}
  T(\xi) := \left|P(\xi)-\mu\right|,
\end{align}
where $P$ is the symbol of the standard Laplacian or the
Molchanov--Vainberg Laplacian and $\mu\in [-1,1]\setminus Z$ is
the Fermi energy (we took $\mu=1$ in the continuum).

\subsection{Number of eigenvalues below a threshold}
\label{s:ltlattice}

We now generalize the results of Section \ref{s:classiclt}
to $T-V$ in $\ell^2(\Z^d)$.
Our assumptions on $T$ are the same as in Assumption \ref{assumt}
with the exception that the ellipticity assumption (3) there is not
needed here.
We recall that item~(4) with $r=(d-1)/2$ in Assumption \ref{assumt}
holds for the BCS operator in \eqref{eq:deflatticebcs} in $d=2$ and
its analog where $-\Delta$ is replaced by $-\Delta_\MV$ in all $d\geq2$.

\begin{theorem}
  \label{tracerestrictionlattice}
  Let $e>0$ and suppose $T(\xi)$ satisfies (1) and (2) in Assumption \ref{assumt}.
  \begin{enumerate}
  \item Let $m\geq1$. If $V\in\ell^m(\Z^d)$, then there exists a constant
    $c_{S}>0$ (which also depends on $d,\tau,m$) such that
    \begin{align}
      \label{eq:tracerestrictionlattice1}
      \|BS(e)\|_{\cs^m(\ell^2(\Z^d))}^m
      \leq c_{S} \min\{e^{1-m},e^{-m}\} \|V\|_{\ell^m(\Z^d)}^m\,.
    \end{align}
    
  \item Suppose $T$ also satisfies (4) in Assumption \ref{assumt} with
    $r\in(0,(d-1)/2]$.
    Let $q\in[1,r+1]$ and $m=\sigma(q,r)$.
    If $V\in \ell^q(\Z^d)=\ell^m\cap\ell^q(\Z^d)$, then there is a constant
    $c_S>0$ (which also depends on $d,\tau,m,q,r$) such that
    \begin{align}
      \label{eq:tracerestrictionlattice4}
      \begin{split}
        & \|BS(e)\|_{\cs^m(\ell^2(\Z^d))}^m
        \leq c_{S}\|V\|_q^m \min\{\log(2+1/e),1/e\}^m\,.
      \end{split}
    \end{align}    
    In particular,
    \begin{align}
      \label{eq:tracerestrictionlattice3}
      \begin{split}
        & \|BS(e)\|_{\cs^m(\ell^2(\Z^d))}^m
        \leq c_S\left[ (\log(2+1/e))^m\|V\|_q^m\theta(1-e) + e^{-m}\|V\|_m^m\theta(e-1)\right]\,.
      \end{split}
    \end{align}
  \end{enumerate}
\end{theorem}

\begin{proof}
  The proofs of \eqref{eq:tracerestrictionlattice1} and
  \eqref{eq:tracerestrictionlattice4} are exactly the same as
  those of \eqref{eq:tracerestriction1} and \eqref{eq:tracerestriction2}
  in the continuum case with two exceptions.
  Due to the absence of high energies in the estimate involving the
  Kato--Seiler--Simon inequality, any $m\geq1$ becomes admissible and the
  $e^{d/s}$-factors for $e>1$ are absent.
  Secondly, by the nestedness of the $\ell^p$ spaces, we may estimate
  $\|V\|_{m} \lesssim \|V\|_{q}$ since $\sigma(q,r)\geq q$
  (cf.~(1) in Remark \ref{remfranksabingen}) to dispose of $\|V\|_m$-norms.
  Estimate \eqref{eq:tracerestrictionlattice3} follows from 
  \eqref{eq:tracerestrictionlattice1}-\eqref{eq:tracerestrictionlattice4}.
\end{proof}

\subsection{Sums of powers of eigenvalues}
The previous estimates allow us to prove an analog of Theorem \ref{classiclt} for
the lattice Schr\"odinger operators considered here.

\begin{theorem}
  \label{classicltlattice}
  Suppose $T(\xi)$ satisfies (1) and (2) in Assumption \ref{assumt}.
  \begin{enumerate}
  \item If $\gamma>0$ and $V\in \ell^{\gamma+1}(\Z^d)$, then there exists a
    constant $c_S>0$ (which also depends on $d,\tau,\gamma$) such that
    \begin{align}
      \label{eq:classicltlattice1}
      \tr_{\ell^2(\Z^d)}(T(-i\nabla)-V)_-^\gamma
      \leq c_{S} \sum_{x\in\Z^d} V_+(x)^{\gamma+1}\,.
    \end{align}

  \item Suppose $T$ also satisfies (4) in Assumption \ref{assumt} with
    $r\in(0,(d-1)/2]$.
    Let $q\in[1,r+1]$ and $m=\sigma(q,r)$.
    Suppose $\delta\in[0,m]$, $\gamma>\delta$, and
    $V\in \ell^{m+\gamma-\delta}\cap\ell^q(\Z^d)$.
    Then $q<m+\gamma-\delta$ and there is a constant
    $c_S$ (which also depends on $d,\tau,m,q,\delta,\gamma,r$) such that
    \begin{align}
      \label{eq:classicltlattice2}
      \tr_{\ell^2(\Z^d)}(T(-i\nabla)-V)_-^\gamma
      \leq c_{S} (\|V_+\|_q^m + \|V_+\|_{m+\gamma-\delta}^{m+\gamma-\delta})\,.
    \end{align}
  \end{enumerate}
\end{theorem}

For $\delta=m-1$, the bound in \eqref{eq:classicltlattice2} restores
$\|V_+\|_{\gamma+1}^{\gamma+1}$ in \eqref{eq:classicltlattice1}.

\begin{proof}
  By the variational principle we can assume $V=V_+\geq0$.
  The proof of \eqref{eq:classicltlattice1} is the same as that of
  \eqref{eq:classiclt} and we omit it.
  To prove \eqref{eq:classicltlattice2} we use
  \eqref{eq:tracerestrictionlattice3} in Theorem
  \ref{tracerestrictionlattice} and $m\geq q$ to bound
  \begin{align*}
    N_e(V) & \lesssim_S \log(2+\frac1e)^m\theta(1-e)\left(\sum_{x\in\Z^d}(V(x)-\frac{e}{2})_+^{m}\right)^{\frac{m}{q}} + e^{-m}\theta(e-1)\sum_{x\in\Z^d}(V(x)-\frac{e}{2})_+^{m}\\
    & \lesssim \log(2+1/e)^m\theta(1-e) \|V\|_q^m + e^{-\delta}\sum_{x\in\Z^d}(V(x)-\frac{e}{2})_+^{m}
  \end{align*}
  for any $0\leq\delta\leq m$. For $\gamma>\delta$ the second
  term on the right contributes with
  \begin{align*}
    \int_0^\infty de\ e^{\gamma-1-\delta} \sum_x (V(x)-e/2)_+^m
    \lesssim_{m,\delta,\gamma} \|V\|_{m+\gamma-\delta}^{m+\gamma-\delta}\,,
  \end{align*}
  whereas the first term contributes with
  \begin{align*}
    \int_0^1 de\ e^{\gamma-1} \log(2+1/e)^m \|V\|_q^m
    \lesssim_{m,q,\gamma} \|V\|_q^m
  \end{align*}
  to the left side of \eqref{eq:classicltlattice2}.
  This concludes the proof.
\end{proof}

\begin{remark}
  Bach, Lakaev, and Pedra \cite{Bachetal2018} proved CLR bounds in $d\geq3$
  when the symbol $T\in C^2(\T^d)$ is a Morse function, i.e., it satisfies
  $T(\xi)\sim|\xi-\xi_0|^2$ near a minimum $\xi_0\in\T^d$.
  This is needed \cite[p.~21]{Bachetal2018} to apply \cite[Theorem 3.2]{Frank2014}
  when computing $\int_{T^{-1}((0,E])}T(\xi)^{-1}\,d\xi$.
\end{remark}

\subsection{Sums of logarithms of eigenvalues}

\begin{theorem}
  \label{logltlattice}
  Let $H=T-V$ in $\ell^2(\Z^d)$ with $T$ satisfying
  (1), (2), and (4) in Assumption \ref{assumt}
  with $r\in(0,(d-1)/2]$.
  Let $q\in[1,r+1]$, $m=\sigma(q,r)$, and
  $V\in \ell^m\cap \ell^q(\Z^d)=\ell^q(\Z^d)$.
  Then for any
  $\gamma>m$ there is a constant $c_S$ (which also depends on $d,s,m,q,\gamma,r$)
  such that
  \begin{align}
    \label{eq:logltfinallattice}
    \sum_j \left(\frac{1}{|\log(\langle 1/e_j\rangle)|}\right)^\gamma
    \leq c_S (\|V_+\|_m^m + \|V_+\|_q^{m})
    \lesssim c_S \|V_+\|_q^{m}\,.
  \end{align}
\end{theorem}

\begin{proof}
  Without loss of generality let $V=V_+$.
  Using \eqref{eq:noose}, the representation \eqref{eq:logltaux}, and
  \eqref{eq:tracerestrictionlattice3} in Theorem
  \ref{tracerestrictionlattice} for $\gamma>0$ and $m\geq1$, i.e.,
  \begin{align}
    N_{r}(V) \lesssim \|V\|_m^m + \|V\|_q^m\cdot \left(\log\left(2+\frac{1}{r}\right)\right)^{m}\,,
  \end{align}
  lets us proceed as in the proof of Theorem \ref{loglt}.
\end{proof}

\begin{remark}
  We make an observation similar to that after Theorem \ref{loglt}.
  The right side of \eqref{eq:logltfinallattice} is bounded by a
  constant times $\|V_+\|_q^{m}$ which is consistent with the right
  side of \eqref{eq:logltprelim} when $r=(d-1)/2$, $m=d+1$, and $q=(d+1)/2$.
  However, we need to restrict ourselves again to $\gamma>d+1$ which
  makes \eqref{eq:logltfinallattice} weaker compared to
  \eqref{eq:logltprelim}. A similar question arises whether
  \eqref{eq:logltfinallattice} can hold for $\gamma=m$.
\end{remark}

\subsection{A CLR bound for powers of the BCS operator in $\ell^2(\Z^d)$}

Let $d\geq1$.
We generalize Theorem \ref{clrbcs} to $|\Delta+\mu|^{1/s}$ with $s>1$ and
$\mu\in[-1,1]\setminus Z$ on $\ell^2(\Z^d)$.
To that end we use that Cwikel's estimate continues to hold in $\ell^2(\Z^d)$.
This is a consequence of an abstract theorem by Birman, Karadzhov, and
Solomyak \cite[Theorem 4.8]{Birmanetal1991}, which also includes an extension
of the Kato--Seiler--Simon inequality.
Recall that the discrete unitary Fourier transform
$\F:\ell^2(\Z^d)\to L^2(\T^d)$ obeys $\|\F\|_{\ell^1(\Z^d)\to L^\infty(\T^d)}\leq1$.
Adapted to our setting, their result reads as follows.

\begin{theorem}[{\cite[Theorem 4.8]{Birmanetal1991}}]
  \label{frankcwikel}
  Let $q>2$, $f\in \ell^q(\Z^d)$, and $g\in L^{q,\infty}(\T^d)$. Then
  \begin{align}
    \label{eq:frankcwikel}
    \|f\F^* g\|_{\cs^{q,\infty}(L^2(\T^d)\to\ell^2(\Z^d))}
    \lesssim_q \|f\|_{\ell^q(\Z^d)} \|g\|_{L^{q,\infty}(\T^d)}\,.
  \end{align}
\end{theorem}
In combination with \eqref{eq:boundnumber} (as in the proof of
Theorem \ref{clrbcs}), \eqref{eq:frankcwikel} yields

\begin{theorem}
  \label{fracbcsclrlattice}
  Let $d\geq1$, $\mu\in[-1,1]\setminus Z$, $\sigma>1$, $p\in(1,\sigma]$,
  and $T_\mu(\xi)$ be defined as in \eqref{eq:deflatticebcs} with the ordinary
  Laplace operator.
  Then $T_\mu(\xi)^{-1/(2\sigma)}\in L^{2p,\infty}(\T^d)$ (not necessarily uniformly
  in $\mu,\sigma,p,d$). Moreover, the number of negative eigenvalues of
  $(T_\mu)^{1/\sigma}-V$ is bounded by a constant (possibly depending on
  $\mu,\sigma,p,d$) times $\|V_+\|_{\ell^{p}(\Z^d)}^{p}$.
\end{theorem}

\begin{proof}
  The proof is analogous to that of Theorem \ref{clrbcs}.
  The bound for the number of negative eigenvalues follows from the
  variational principle, \eqref{eq:boundnumber}, and \eqref{eq:frankcwikel}
  (with $f(x)=|V(x)|^{1/2}$ and $g(\xi)=(T_\mu(\xi))^{-1/(2\sigma)}$ for $x\in\Z^d$
  and $\xi\in\T^d$).
  Thus, we are left with showing $T_\mu(\xi)^{-1/(2\sigma)}\in L^{2p,\infty}(\T^d)$
  with $p\leq \sigma$. Since $|\T^d|=1$, it suffices to check
  \begin{align}
    \label{eq:discreteweaknorm}
    |\{\xi\in\T^d:\, T_\mu(\xi)\leq\beta^{2\sigma}\}|
    \lesssim_{\mu,d,\sigma,p} \beta^{2p} \quad \text{for}\ \beta\leq1\,.
  \end{align}
  Since $\mu$ is a given, fixed parameter, we may even suppose
  $\beta^{2\sigma}<1-\mu$ in the following. Then $T_\mu(\xi)\leq\beta^{2\sigma}$
  is equivalent to the bounds
  \begin{align}
    \label{eq:fracbcsclrlatticeaux}
    -\beta^{2\sigma} \leq d^{-1}\sum_{j=1}^d\cos(2\pi\xi_j) - \mu \leq \beta^{2\sigma}\,,
    \quad \xi_j\in(-\frac12,\frac12)\,.
  \end{align}
  Since $1-x^2/2\leq\cos x\leq1-x^2/(2\pi)$ for all $x\in(-\pi,\pi)$,
  \eqref{eq:fracbcsclrlatticeaux} implies
  \begin{align*}
    1-\mu-\beta^{2\sigma} \leq \frac{2\pi^2}{d}|\xi|^2 \leq 1-\mu+\beta^{2\sigma}\,.
  \end{align*}  
  Thus, the left side of \eqref{eq:discreteweaknorm} is bounded from
  above by $\beta^{2\sigma} \leq \beta^{2p}$
  since $p\leq \sigma$ and $\beta<1$.
  This concludes the proof.
\end{proof}

\appendix
\section{Alternative proof of Theorem \ref{classiclt}}
We now give an alternative proof of Theorem \ref{classiclt} (1) for $\gamma>0$
using an observation made by Frank \cite[p.~794]{Frank2009}, together with Theorem \ref{clrbcs}.

\begin{theorem}
  \label{classicltrumin2}
  Suppose $T(\xi)$ satisfies (1)-(3) in Assumption \ref{assumt}.
  If $\gamma>0$ and $V\in L^{\gamma+1}\cap L^{\gamma+d/s}(\R^d)$
  then there exists a constant $c_S>0$ (which also depends on $d,s,\gamma$)
  such that
  \begin{align}
    \label{eq:classicltrumin}
    \tr_{L^2(\R^d)}(T(-i\nabla)-V)_-^\gamma
    \leq c_S \int_{\R^d} (V_+(x)^{\gamma+1} + V_+(x)^{\gamma+d/s})\,dx\,.
  \end{align}
\end{theorem}

\begin{proof}
  Without loss of generality we assume $V\geq0$.
  For $E>0$ and $\sigma>1$ we record
  \begin{align}
    T(-i\nabla) + E \geq c_\sigma\cdot T(-i\nabla)^{1/\sigma} \cdot E^{1/\sigma'} 
  \end{align}
  for $\sigma'=(1-1/\sigma)^{-1}$ and some $c_\sigma>0$.
  This observation and Theorem \ref{clrbcs} imply that the number of
  eigenvalues $N(2E,T-V)$ of $T-V$ below $-2E<0$ is bounded by
  \begin{align}
    \begin{split}
      N(2E,T-V)
      & = N(0,T+E-(V-E))
      \leq N(0,c_\sigma E^{1/\sigma'}T^{1/\sigma}-(V-E))\\
      & = n(1,|V-E|^{\frac12}(c_\sigma E^{1/\sigma'}T^{1/\sigma})^{-1}(V-E)^{\frac12})\\
      & = N(0,c_\sigma T^{1/\sigma}-E^{-1/\sigma'}(V-E))\\
      & \lesssim_\sigma E^{-\frac{\sigma}{\sigma'}}\|(V-E)_+\|_{L^{\sigma}(\R^d)}^\sigma + E^{-\frac{\sigma d/s}{\sigma'}} \|(V-E)_+\|_{L^{\sigma d/s}(\R^d)}^{\sigma d/s}\,.
    \end{split}
  \end{align}
  Thus, we obtain for any $\gamma>d\sigma/(s\sigma')$,
  \begin{align}
    \begin{split}
      \tr(T-V)_-^\gamma
      & = \int_0^\infty dE\, E^{\gamma-1}\cdot N(E,T-V)\\
      & \lesssim \int_0^\infty dE\, \left[E^{\gamma-1-\frac{\sigma}{\sigma'}}\int_{\R^d}dx\, (V(x)-\frac{E}{2})_+^\sigma + E^{\gamma-1-\frac{d\sigma}{s\sigma'}}\int_{\R^d}dx\, (V(x)-\frac{E}{2})_+^{\frac{d\sigma}{s}}\right]\\
      & \sim \int_{\R^d}(V(x)^{\gamma+1} + V(x)^{\gamma+d/s})\,dx\,.
    \end{split}
  \end{align}
  This concludes the proof.
\end{proof}

\begin{remarks}
  \begin{enumerate}
  \item We do not know whether the CLR bounds in Theorem \ref{clrbcs} and an
    argument similar to that in the proof of Theorem \ref{classicltrumin2} can
    be used to prove estimates for sums of logarithms of eigenvalues as in
    Theorem \ref{loglt}.

  \item Theorem \ref{classicltrumin2} for $\gamma=1$ can be proved using
    Rumin's method, see also \cite[Proposition~4]{Franketal2021} or
    \cite[Section~6]{Frank2021}. The case $\gamma>1$ then follows from this
    together with the argument of Aizenman and Lieb \cite{AizenmanLieb1978} and
    the observation
    \begin{align*}
      \int_{\R^d}\left(T(\xi)-V(x)\right)_-^\gamma\,d\xi
      & = \int_0^{V(x)}dt\, (V(x)-t)^\gamma \int_{S_t}\frac{d\Sigma_{S_t}(\xi)}{|\nabla P(\xi)|}\\
      & \sim \int_0^{V(x)}dt\, (V(x)-t)^\gamma \cdot (1+t)^{d/s-1}
        \sim V(x)^{\gamma+1} + V(x)^{\gamma+d/s}\,.
    \end{align*}
  \end{enumerate}
\end{remarks}

\section*{Acknowledgments}
We are grateful to Volker Bach for valuable discussions
and to Kouichi Taira for providing helpful comments.
Special thanks go to Rupert Frank for providing critical remarks
on Theorems \ref{classiclt} and \ref{clrbcs},
and for pointing out that the methods of Rumin and \cite{Frank2009,Frank2014}
provide an alternative proof of Theorem \ref{classiclt} (1).


\def\cprime{$'$}

\end{document}